\documentclass[10pt, conference, letterpaper]{IEEEtran}
\IEEEoverridecommandlockouts
\usepackage{amsmath,amssymb,amsfonts}
\usepackage{amsthm}
\usepackage[linesnumbered,ruled,vlined]{algorithm2e}
\usepackage{graphicx}
\usepackage{authblk}
\usepackage{textcomp}
\usepackage{xcolor}
\usepackage[font={footnotesize}]{subcaption}
\usepackage{cite}
\def\BibTeX{{\rm B\kern-.05em{\sc i\kern-.025em b}\kern-.08em
    T\kern-.1667em\lower.7ex\hbox{E}\kern-.125emX}}

\newtheorem{theorem}{Theorem}
\newtheorem{corollary}{Corollary}
\newtheorem{lemma}{Lemma}
\newtheorem{assumption}{Assumption}
\usepackage{mathtools}
\usepackage{soul}

\SetKwInput{KwConstants}{Constants}

\SetKwInput{KwOutput}{Output}
\SetKwFor{ClientFor}{each client}{in parallel:}{}
\SetKwFor{SubnetFor}{each subnet}{in parallel:}{}
\SetKwFor{AggrFor}{all sampled clients}{aggregate:}{}

\begin{document}
\title{Taming Subnet-Drift in D2D-Enabled Fog Learning: A Hierarchical Gradient Tracking Approach}
\author{Evan Chen\IEEEauthorrefmark{1}, Shiqiang Wang\IEEEauthorrefmark{2}, and Christopher G. Brinton\IEEEauthorrefmark{1} \\
\IEEEauthorrefmark{1}School of Electrical and Computer Engineering, Purdue University, West Lafayette, IN 47907 \\ \IEEEauthorrefmark{2}IBM T. J. Watson Research Center, Yorktown Heights, NY 10598 \\
Email: \{chen4388, cgb\}@purdue.edu, wangshiq@us.ibm.com
\thanks{This work was supported by the National Science Foundation (NSF) under grants CPS-2313109 and CNS-2212565, by DARPA under grant D22AP00168, and by the Office of Naval Research (ONR) under grant N000142212305.}}


\maketitle

\begin{abstract}
Federated learning (FL) encounters scalability challenges when implemented over fog networks. Semi-decentralized FL (SD-FL) proposes a solution that divides model cooperation into two stages: at the lower stage, device-to-device (D2D) communications is employed for local model aggregations within subnetworks (subnets), while the upper stage handles device-server (DS) communications for global model aggregations. However, existing SD-FL schemes are based on gradient diversity assumptions that become performance bottlenecks as data distributions become more heterogeneous. In this work, we develop semi-decentralized gradient tracking (SD-GT), the first SD-FL methodology that removes the need for such assumptions by incorporating tracking terms into device updates for each communication layer. Analytical characterization of SD-GT reveals convergence upper bounds for both non-convex and strongly-convex problems, for a suitable choice of step size. We employ the resulting bounds in the development of a co-optimization algorithm for optimizing subnet sampling rates and D2D rounds according to a performance-efficiency trade-off. Our subsequent numerical evaluations demonstrate that SD-GT obtains substantial improvements in trained model quality and communication cost relative to baselines in SD-FL and gradient tracking on several datasets.
\end{abstract}


\section{Introduction}
\noindent Federated learning (FL) has emerged as a promising technique for distributed machine learning (ML) over networked systems \cite{kairouz2021advances,li2020federated}. FL aims to solve problems of this form:
\begin{align}
    \min_{x\in \mathbb{R}^d}f(x) &= \frac{1}{n}\sum_{i=1}^n f_i(x)\\
    \textrm{where } f_i(x) &= \mathbb{E}_{\xi_i \sim \mathcal{D}_i}f_i(x;\xi_i),
\end{align}
where $n$ is the total number of clients (typically edge devices) in the system, $f_i(x)$ is the local ML loss function computed at client $i$ for model parameters $x \in \mathbb{R}^d$, $\mathcal{D}_i$ is the local data distribution at client $i$, and $\xi_i$ is a random sample from $\mathcal{D}_i$. 

Conventionally, FL employs a two-step iterative algorithm to solve this optimization: (i) \textit{local model update}, where gradient information computed on the local device dataset is used to update the local model, and (ii) \textit{global model aggregation}, where a central server forms a consensus model across all devices. In wireless networks, however, device-server (DS) communications for the global aggregation step can be expensive, especially for large ML models over long DS distances. Much research in FL has been devoted to improving this communication efficiency, with typical approaches including model sparsification/quantization \cite{wang2022federated,amiri2020federated,li2021talk}, device sampling~\cite{wang2021device}, and aggregation frequency minimization \cite{karimireddy2020scaffold,mishchenko2022proxskip}.

Recent research has considered how decentralizing FL's client-server star topology can improve communication overhead, e.g., by introducing more localized communications wherever possible. In the extreme case of severless FL, model aggregations are conducted entirely through short range device-to-device (D2D) communications \cite{koloskova2019decentralized,lian2017can,zehtabi2022decentralized}. More generally, \textit{fog learning} \cite{hosseinalipour2020federated,nguyen2022fedfog,hosseinalipour2022multi} considers distributing FL over fog computing architectures, where a hierarchy of network elements between the edge and cloud enables horizontal (i.e., intra-layer) and vertical (i.e., inter-layer) communications.


\begin{figure}
\centerline{\includegraphics[width=0.5\textwidth]{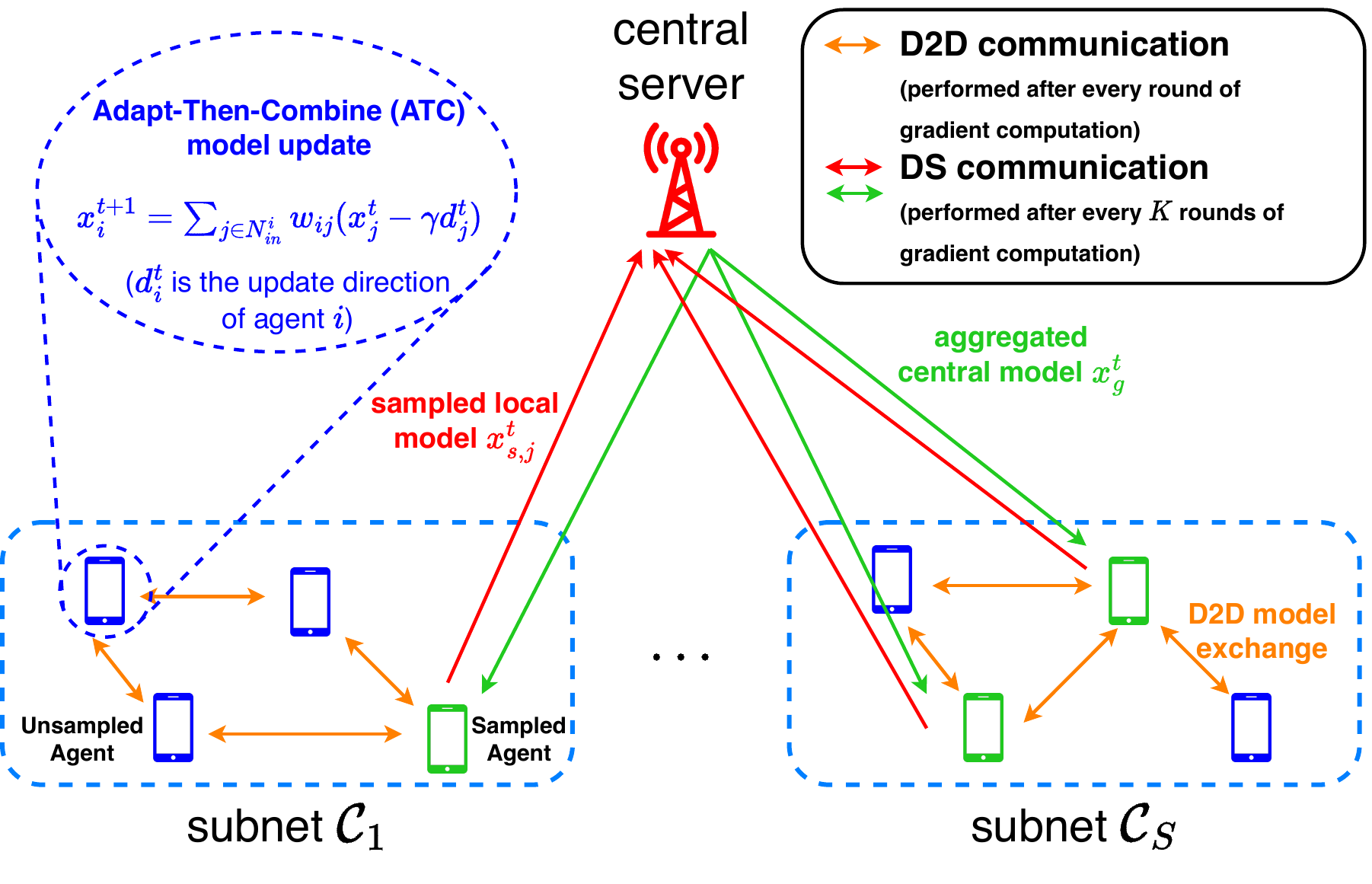}}
\caption{Illustration of semi-decentralized FL. Clients in each subnet communicate via iterative low-cost D2D communications to conduct local aggregations. Once they have converged towards a consensus within the subnet, the central server conducts a global aggregation across sampled devices using DS communication.\vspace{-0.15in}}
\label{fig1}
\end{figure}

\subsection{Semi-Decentralized FL and Subnet-Drift}
\textit{Semi-decentralized FL} (SD-FL) has emerged as an important implementation of fog learning~\cite{lin2021semi,yemini2022semi}. Its overall architecture is depicted in Fig.~\ref{fig1}. Devices are grouped into subnetworks (subnets) of close physical proximity, according to their ability to form D2D connections. For example, consider a set of 5G mobile devices in a cell aiming to learn an ML model: peer relationships can establish D2D-enabled subnets, with the main server located at the base station \cite{suraci2021trusted}.

To enhance communication efficiency, the model aggregation in SD-FL is conducted in two stages: (i) \textit{iterative cooperative consensus formation} of local models within subnets, and (ii) \textit{DS communication among sampled devices} for global aggregation. The idea is that frequent, low-cost in-subnet model aggregations should reduce the burden placed on global, cross-subnet aggregations, as they can occur less frequently engaging fewer clients. However, a fundamental challenge in SD-FL is \textit{managing gradient diversity across subnets:} the more local aggregations we perform, the more the global model drifts away from the global optimum towards a linear combination of local optimums of each subnet. This ``subnet-drift" manifests from the client-drift problem in FL, due to non-i.i.d. local datasets across clients \cite{karimireddy2020scaffold,liu2023decentralized,mishchenko2022proxskip}. 

In this work, we are interested in addressing the subnet-drift challenge for SD-FL. Although some existing works alleviate client drift by letting clients share a portion of their datasets with their neighbors and/or the server \cite{tu2020network,wang2021device,zhao2018federated}, such approaches present privacy issues that FL aims to avoid. As a result, we turn to concepts in \textit{gradient tracking}, which have been successful in mitigating data heterogeneity challenges in fully decentralized learning and do not require data sharing\cite{koloskova2019decentralized,lian2017can}. In this respect, the hierarchical nature of SD-FL presents two key research challenges. First, the differing timescales of D2D and DS communications need careful consideration on how a client should employ gradient information from the server versus from its neighbors. Second, randomness in client participation for DS communication may create bias in aggregated gradient information. We thus pose the following question: 
\begin{center}
\textit{How do we alleviate subnet drift in semi-decentralized FL through gradient tracking while ensuring gradient information is well mixed throughout the system?}
\end{center}
To address this, a key component of our design is to consider two separate gradient tracking terms, one for each communication stage of SD-FL, which treat the incoming information differently.
Our convergence analysis and subsequent experiments demonstrate how this stabilizes the global learning process.

\subsection{Outline and Summary of Contributions}
\begin{itemize}
    \item We propose Semi-Decentralized Gradient Tracking (SD-GT), the first work to integrate gradient tracking into SD-FL which is robust to data heterogeneity. SD-GT can tolerate a large number of D2D communications between two global aggregation rounds without risking convergence to a sub-optimal solution (Sec. \ref{sec:III}).
    
    \item We conduct a Lyapunov-based convergence analysis, showing convergence upper bounds for both non-convex and strongly convex functions. We employ these results in a co-optimization of convergence speed and communication efficiency based on the D2D communication rounds and subnet sampling rates (Sec. \ref{sec:IV}).

    \item Our experiments verify that SD-GT obtains substantial improvements in trained model quality and convergence speed relative to baselines in SD-FL and gradient tracking. Moreover, we verify the behavior of our co-optimization optimization in adapting to the relative cost of D2D vs. DS communications (Sec. \ref{sec:V}).
    
\end{itemize}

\section{Related Works}
\label{sec:II}


\textbf{Hierarchical FL.} Hierarchical FL has received considerable attention for scaling up model training across large numbers of edge devices. Most of this work has considered a multi-stage tree extension of FL \cite{wang2021resource,liu2020client,hosseinalipour2022multi,wang2022infedge}, i.e., with each stage forming its own star topology for local aggregations. A commonly considered use case has been the three-tier hierarchy involving device, base station, and cloud encountered in cellular networks. Optimized hierarchical aggregations have demonstrated significant improvements in convergence speed and/or communication efficiency. In a separate domain, these concepts have been employed for model personalization in cross-silo FL~\cite{zhou2023hierarchical}.

Our work focuses specifically on semi-decentralized FL, where edge subnets conduct local aggregations via D2D-enabled cooperative consensus formation   \cite{lin2021semi,yemini2022semi,parasnis2023connectivity}. SD-FL is intended for settings where DS communications are costly, e.g., due to long distances. The authors of \cite{lin2021semi} were the first to formally study the convergence behavior of SD-FL, wherein they proposed a control algorithm to maintain convergence based on approximations of data-related parameters. \cite{parasnis2023connectivity} developed SD-FL based on more general models of subnet topologies that may be time-varying and directed. 
A main issue with all current SD-FL papers is that they assume that either the gradient, gradient diversity, or data-heterogeneity are bounded, while \cite{lin2021semi} and  \cite{parasnis2023connectivity} even require knowledge on the connectivity of each subnet. In our work, we are able to remove the requirement of any knowledge on data distribution or subnet topology and still guarantee convergence.

\textbf{Gradient tracking for communication efficiency.} Gradient tracking (GT) methods \cite{di2016next,nedic2017achieving,tian2018asy,koloskova2021improved,sun2022distributed} were proposed to mitigate data heterogeneity in decentralized optimization algorithms. The main idea is to track the gradient information from neighbors every time a communication is performed. GT has become particularly popular in settings where communication costs are high, as it enables algorithms to reach the optimum point using a large number of local updates and minimal communications\cite{karimireddy2020scaffold,mishchenko2022proxskip,liu2023decentralized,alghunaim2023local,ge2023gradient,berahas2023balancing,zhang2021low}. Assumptions on data heterogeneity can be lifted under proper initialization of gradient tracking variables.

Our work instead considers GT under a semi-decentralized network setting. In this respect, \cite{huang2022tackling} discussed GT under a hierarchical network structure, where they assumed: (i) random edge activation within subgraphs and (ii) all subgraphs being connected by a higher layer graph that communicates after every gradient update. However, the hierarchical structure they consider is different from SD-FL, where in our setting D2D communication usually is cheaper than DS communication and thus occurs at a much higher frequency. In this paper, we develop a GT methodology that is aware of the diversity in information mixing speeds between D2D and DS, and track this difference by maintaining two separate GT terms.

\section{Proposed Method}
\label{sec:III}
\begin{figure}
\centerline{\includegraphics[width=0.5\textwidth]{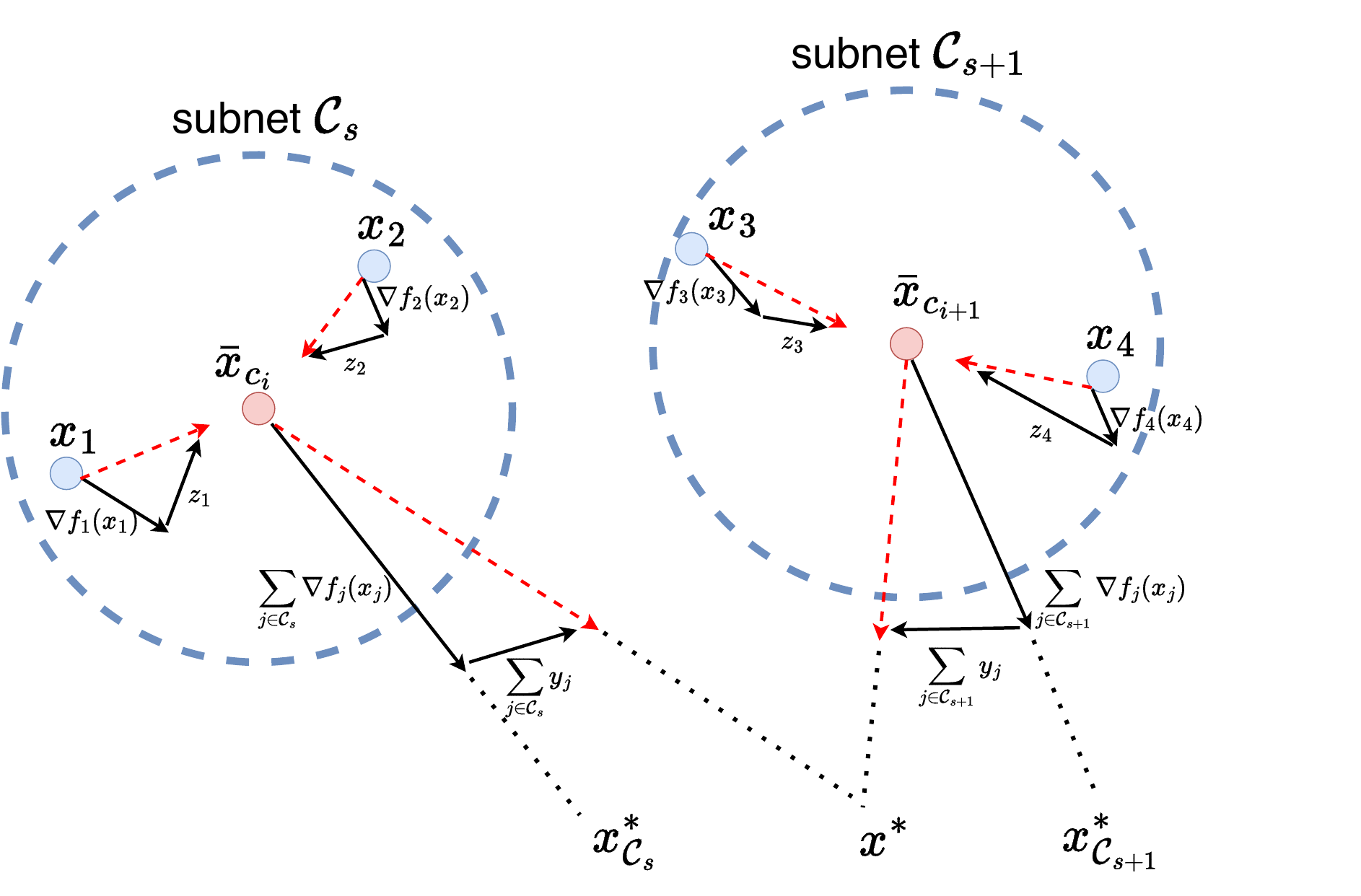}}
\caption{An illustration of how SD-GT deals with subnet-drifting. With the introduction of in-subnet GT term $z_i^t$, all clients within each subnet are able to converge towards a consensual location of the subnet. And the inter-subnet GT term $y_i^t$ corrects the update direction of the whole subnet so that it no longer converges towards the optimal solution $x_{\mathcal{C}_s}^* $of the subnet $\mathcal{C}_s$ but the optimal solution $x^*$ of the whole network.\vspace{-0.15in}}
\label{fig2}
\end{figure}
In this section, we first introduce the overall network structure of SD-FL (Sec. \ref{sec:IIIA}). Then we develop our SD-GT algorithm, explaining the usage of each tracking variable and how they solve the subnet-drift problem (Sec. \ref{sec:IIIB}). Finally, we show that our method encapsulates two existing methods under specific network topologies (Sec. \ref{sec:IIIC}). 

\subsection{Network Model and Timescales}
\label{sec:IIIA}
We consider a network containing a central server connected upstream from $n$ clients (edge devices), indexed $i = 1,...,n$. As shown in Figure \ref{fig1}, the devices are partitioned into $S$ disjoint subnets 
$\mathcal{C}_1, \ldots, \mathcal{C}_S$. Subnet $s$ contains $m_{s} = |\mathcal{C}_s|$ clients, where $\sum_{s=1}^S m_{s} = n$. Similar to existing works in SD-FL~\cite{lin2021semi,yemini2022semi}, we do not presume any particular mechanism by which clients have been grouped into subnets, except that clients within the same subnet are capable of engaging in D2D communications according to a wireless protocol, e.g., devices within a 5G cell.

For every client $i \in \mathcal{C}_s$, we let $\mathcal{N}^{\mathrm{in}}_i \subseteq \mathcal{C}_s$
be the set of input neighbors for D2D transmissions. Considering all clients $i, j \in \mathcal{C}_s$, we define $W_{s} = [w_{ij}] \in \mathbb{R}^{m_{s}\times m_{s}}$ to be the D2D communication matrix for subnet $s$, where $0 < w_{ij} \leq 1$ if $j \in \mathcal{N}^{\mathrm{in}}_i$, and $w_{ij} = 0$ otherwise. As we will see in Sec.~\ref{sec:IIIB}, $w_{ij}$ is the weight that client $i$ will apply to information received from client $j$. We then can define the network-wide D2D matrix $W = \mathrm{diag}(W_{1}, \ldots, W_{S}) \in \mathbb{R}^{n\times n}$, which is block-diagonal given that the subnets do not directly communicate. 
In Sec.~\ref{sec:IV-A}, we will discuss further assumptions on the subnet matrices $W_s$ for our convergence analysis.

The SD-GT training process consists of two timescales. The outer timescale, $t = 1,2,\ldots,T$, indexes global aggregations carried out through DS communications. The inner timescale, $k = 1, \ldots, K$, indexes local training and aggregation rounds carried out via D2D communications. 
We assume a constant $K$ local rounds occur between consecutive global aggregations.
\subsection{Learning Model}
\label{sec:IIIB}
As shown in Algorithm \ref{alg:1}, each client maintains two gradient tracking terms, $y_i^t$ and $z_i^t$, which track (i) the gradient information between different subnets and (ii) the gradient information inside each subnet, respectively. These two variables act as corrections to the local gradients so that the update direction can guarantee convergence towards global optimum as in Figure \ref{fig2}.


\textbf{In-subnet Updates.} We denote $x_i^{t,k}$ as the ML parameter vector stored at client $i$ during the $t$\textsuperscript{th} global aggregation round and $k$\textsuperscript{th} D2D communication round. The local model updates are performed by first updating its local model, and then making a linear combination with models received from its neighbors, also known the Adapt-Then-Combine (ATC) scheme. ATC is known to have a better performance compared to other mixing schemes\cite{tu2012diffusion}. The update direction not only includes the gradient direction computed from the local function $\nabla f_i(\cdot)$ but also the gradient tracking terms $z_i^t$ and $y_i^t$:
\begin{align}
    \textstyle x_{i}^{t,k+\frac{1}{2}} &\textstyle= x_{i}^{t,k} - \gamma \left( \nabla f_i(x_{i}^{t,k}, \xi_i^{t,k}) + y_{i}^{t} + z_i^{t}\right), && \forall i,
    \label{eq3}\\
    \textstyle  x_{i}^{t,k+1} &\textstyle= \sum_{j \in \mathcal{N}^{\mathrm{in}}_i\cup \{i\}} w_{ij}x_{j}^{t,k+\frac{1}{2}}, && \forall i,
    \label{eq4}
\end{align}
where $\nabla f_i(x_{i}^{t,k}, \xi_i^{t,k})$ denotes the stochastic gradient of $\nabla f_i(x_{i}^{t,k})$.
After the $K$ rounds of D2D communications, each client updates its own in-subnet gradient tracking term $z_i^t$ using the measured difference of the variable $x_i^{t,k}$ every time a D2D communication within the subnet is performed:
\begin{align}
    \textstyle \Tilde{z}_i^{t,k}  &\textstyle= x_{i}^{t,k+\frac{1}{2}} - x_{i}^{t,k} + \gamma y_i^t, && \forall i,
    \label{eq5}\\
    \textstyle z_i^{t+1} &\textstyle= z_i^t + \frac{1}{K\gamma}\sum_{k=1}^K(\Tilde{z}_i^{t,k} - \sum_{j \in \mathcal{N}^{\mathrm{in}}_i \cup \{i\}} w_{ij}\Tilde{z}_j^{t,k}), && \forall i.\label{eq6}
\end{align}

\textbf{Global Aggregation.} The central server is able to choose the total number of clients $h_{s} \in \{1,...,m_s\}$ to sample from each subnet $s$, e.g., based the DS communication budget. These will be variables in our optimization considered in Sec.~\ref{sec:IV-C}. Those clients that are not sampled by the central server for round $t$ will not update their parameters, and maintain $x_i^t$ and $y_i^t$ into the next communication round. 

We denote $x_\mathrm{g}^t$ as the global model that is stored at the server. At each global aggregation $t$, the server accumulates the gradient information collected from each subnet, and updates the inter-subnet gradient tracking terms that are stored on the server, which we denote $\psi_{s}^t$ for subnet $s$.

Now, for each subnet $s$, let the sampled clients for round $t$ be indexed by $j = 1,\ldots, h_{s}$. Before uploading the information to the server, each client conducts a cancellation of the inter-subnet gradient tracking information so that the global model receives unbiased gradient information. Formally, we have the following sequence of updates, based on intermediate quantities $\textstyle\Tilde{x}_{s,j}^{t+1}$ at the client-side and $\textstyle\Tilde{x}_\mathrm{g} ^{t+1}$ at server-side:
\begin{align}
        &\textstyle\Tilde{x}_{s,j}^{t+1} = x_{{s,j}}^{t, K+1} - x_{{s,j}}^{t, 1}+ K\gamma y_{{s,j}}^t,&& \forall j,\forall s, \label{eq7}\\
        &\textstyle\Tilde{x}_\mathrm{g} ^{t+1} = \frac{1}{S\cdot h_{s}}\sum_{s = 1}^{S}\sum_{j=1}^{h_{s}}  \Tilde{x}_{{s,j}}^{t+1}, \label{eq8}\\
        &\textstyle x_\mathrm{g} ^{t+1} = x_\mathrm{g} ^t + \Tilde{x}_\mathrm{g}^{t+1},\\
        &\textstyle \psi_{{s}}^{t+1} = \frac{1}{K\gamma}(\frac{1}{h_{s}}\sum_{j=1}^{h_{s}}\Tilde{x}_{{s,j}}^{t+1} - \Tilde{x}_\mathrm{g} ^{t+1}), && \forall s.\label{eq10}
\end{align}
\eqref{eq7} is a client-side computation, while \eqref{eq8}-\eqref{eq10} occur at the server. Finally, the server broadcasts the updated global model $x_\mathrm{g}^{t+1}$ and inter-subnet gradient tracking terms $y_{s}^{t+1}$ to the sampled clients to complete the synchronization:
\begin{align*}
        \textstyle x_{{s,j}}^{t,1} &\textstyle= x_\mathrm{g} ^{t+1}, \quad \forall j,\\
        \textstyle y_{{s,j}}^{t+1} &\textstyle= \psi_{s}^{t+1}, \quad \forall j.
\end{align*}

Maintaining two gradient tracking terms is an essential feature of SD-GT for stabilizing convergence. In particular, if we only used the between-subnet measure $y_i^t$ to track gradient information, then the gradient information within each subnet would deviate from the average of the subnet, preventing the system from converging towards the global minimum. On the other hand, if we only used the within-subnet measure $z_i^t$ to track gradient information, then each subnet $s$ will tend to converge towards its local minimum $x_{s}^*$ instead.

\begin{algorithm}[t]
{\footnotesize
\caption{\footnotesize SD-GT: Semi-Decentralized Gradient Tracking}
\label{alg:1}
\KwConstants{step size $\gamma > 0$, initial model parameter $x^0$}
\KwOutput{$x_\mathrm{g} ^t$}
\textbf{Local parameter initialization}\\
$x_i^{1,1} \gets x^0, \quad \forall i$\\
$y_1^1 = y_2^1 = \cdots = y_n^1 \gets 0, \quad \forall i$\\
$z_1^1 = z_2^1 = \cdots = z_n^1 \gets 0, \quad \forall i$\\
\textbf{Server parameter initialization}\\
$x_\mathrm{g} ^1 \gets x^0$\\
$\psi_{s}^1 \gets 0, \quad \forall s$\\
\For{$t\leftarrow 1,\ldots,T$}{
\tcc{Step 1: In-Subnet Model Update}
\For{$k\leftarrow 1,\ldots,K$}{
\ClientFor{$i \gets 1,\ldots,n$}{
Perform in-subnet update on the local model $x_i^{t,k}$ based on \eqref{eq3} and \eqref{eq4}.

}
}
\tcc{Step 2: Local Tracking Term Update}
\ClientFor{$i \gets 1,\cdots,n$}{
    Update the in-subnet gradient tracking variable $z_i^t$ based on \eqref{eq5} and \eqref{eq6}.
}
\tcc{Step 3: Global Aggregation}
sample $h_{s}$ random clients $x_{s,j} \sim \mathcal{C}_s$ from every subnet\\
\SubnetFor{$s\gets 1,\ldots,S$}{
\AggrFor{$j \leftarrow 1,\ldots,h_{s}$}{
Use \eqref{eq7} to \eqref{eq10} to update variables $x_\mathrm{g} ^t$ and $\psi_{s}^t$, $\forall s$ on the server, then broadcast them to all sampled clients.
}
}
\textbf{All clients that are not sampled:}\\
    $x_i^{t+1,1} \gets x_i^{t,K+1}$\\
    $y_i^{t+1} \gets y_i^t$
}
}\end{algorithm}
\subsection{Connection with Existing Methods}
\label{sec:IIIC}
Some existing methods can be shown to be special cases of our algorithm under certain network structures.

\textbf{Case 1: Every client is its own subnet ($S = n$).}
Under this setting, the server always samples the full subnet since each subnet contains only one client. Then we can see that with the initialization $z_1^1 = z_2^1 = \ldots = z_n^1 = 0$, the in-subnet gradient tracking terms are always zero:
\begin{align*}
    \textstyle z_i^{t+1} = z_i^t + \frac{1}{K\gamma}\sum_{k=1}^K(\Tilde{z}_i^{t,k} - \Tilde{z}_i^{t,k}) = z_i^t = 0, \quad \forall i.
\end{align*}
The global gradient tracking term $y_{s,j}^t$ can be formulated as:
\begin{align*}
    \textstyle y_{s,j}^{t+1} = \psi_{s}^{t+1} & \textstyle= \frac{1}{K\gamma}(\sum_{j=1}^{h_{s}} \frac{1}{h_{s}}\Tilde{x}_{{s,j}}^{t+1} - \Tilde{x}_\mathrm{g} ^{t+1})\notag\\
    &\textstyle= y_{s,j}^{t} + \frac{1}{K\gamma}(x_{s}^{t,K+1} - x_\mathrm{g} ^{t+1}), \quad \forall s.
\end{align*}
With the in-subnet gradient tracking term being zero and the global gradient tracking term in the form above, our algorithm aligns with Proxskip~\cite{mishchenko2022proxskip} under this setting.

\textbf{Case 2: Single subnet with one D2D communication round ($S = 1, K = 1$).}
Under this setting, the global gradient tracking term is always zero since $\sum_{j=1}^{h_s} \frac{1}{h_s}\Tilde{x}_{s,j}^{t+1} = \Tilde{x}_\mathrm{g} ^{t+1}$:
\begin{equation*}
    \textstyle \psi_{s}^{t+1} = \psi_{s}^{t} + \frac{1}{\gamma}(\sum_{j=1}^{h_{s}} \frac{1}{h_c}\Tilde{x}_{{s,j}}^{t+1} - \Tilde{x}_\mathrm{g} ^{t+1}) = y_{s}^{t} = 0, \quad \forall s.
\end{equation*}
Since the D2D rounds are set to one, the update of $z_i^t$ for every client $i$ can be formulated as:
\begin{align*}
    \textstyle z_i^{t+1} \textstyle=& z_i^t + \frac{1}{K\gamma}\sum_{k=1}^K(\Tilde{z}_i^{t,k} - \sum_{j \in \mathcal{N}^{\mathrm{in}}_i \cup \{i\}} w_{ij}\Tilde{z}_j^{t,k})\\
    =&\textstyle \sum_{j \in \mathcal{N}^{\mathrm{in}}_i \cup \{i\}} w_{ij}(z_j^t+ \nabla f_j(x_j^t)) - \nabla f_i(x_i^t),\quad\forall i .
\end{align*}
If we define $\hat{z}_i^t = z_i^t + \nabla f_i(x_i^t)$, we are left with:
\begin{align*}
    \textstyle x_i^{t+1} &=\textstyle \sum_{j \in \mathcal{N}^{\mathrm{in}}_i \cup \{i\}} w_{ij}(x_j^{t} - \gamma \hat{z}_i^t), &&\forall i,\\
    \textstyle\hat{z}_i^{t+1} &= \sum_{j \in \mathcal{N}^{\mathrm{in}}_i \cup \{i\}} w_{ij}\hat{z}_j^t + \nabla f_i(x_i^{t+1}) -  \nabla f_i(x_i^t),&&\forall i,
\end{align*}
which aligns with the gradient tracking algorithm \cite{sun2016distributed,di2016next}.

\section{Analysis and Optimization}
\label{sec:IV}
We first show the assumptions used for the proof (Sec. \ref{sec:IV-A}), then provide the convergence analysis under non- and strong-convexity (Sec. \ref{sec:IV-B}). Finally, we derive a co-optimization algorithm that considers the trade-off between communication cost and performance (Sec. \ref{sec:IV-C}). We defer the proofs of supporting lemmas to the appendices.

\subsection{Convergence Analysis Assumptions}
\label{sec:IV-A}
The first three assumptions established are general assumptions \cite{liu2023decentralized,mishchenko2022proxskip,wang2022federated} that are applied to both theorems in this paper, while the last assumption is a stricter condition \cite{tian2018asy} that guarantees a better convergence rate in Theorem \ref{thm2}. 

\begin{assumption}\textbf{(L-smooth)} Each local objective function $f_i$:
$\mathbb{R}^d \rightarrow \mathbb{R}$ is L-smooth:
\begin{equation*}
    \|\nabla f_i(x) - \nabla f_i(y)\|_2 \leq L\|y - x\|_2, \quad \forall x,y\in \mathbb{R}^d.
\end{equation*}

\label{asmp1}
\end{assumption}
\begin{assumption}\textbf{(Mixing Rate)} Each subnet $\mathcal{C}_1, \ldots, \mathcal{C}_S$ has a strongly connected graph, with doubly stochastic weight matrix $W_{s} \in \mathbb{R}^{m_{s} \times m_{s}}$. Define $J_{s} = \frac{\mathbf{11^T}}{m_{s}}$, then there exists a constant $\rho_{s}\in (0,1]$ s.t.
\begin{equation*}
    \| X (W_{s} - J_{s})\|_F^2 \leq (1 - \rho_{s})\| X(I - J_{s})\|_F^2, \quad \forall X \in \mathbb{R}^{d \times m_{s}}.
\end{equation*}
\label{asmp2}
\end{assumption}
\begin{assumption}\textbf{(Bounded Variance)}
Variances of each client's stochastic gradients are uniformly bounded.
\begin{align*}
    \mathbb{E}_{\xi \sim \mathcal{D}_i}\|\nabla f_i(x;\xi) - \nabla f_i(x)\|_2^2 \leq \sigma^2, \forall i \in [1, n], \forall x \in \mathbb{R}^d.
\end{align*}

\label{assmp3}
\end{assumption}
\begin{assumption}\textbf{($\mu$-strongly-convex)}
    Every local objective function $f_i: \mathbb{R}^d \rightarrow \mathbb{R}$ is $\mu$-strongly convex with $0 < \mu \leq L$.
\begin{equation*}
    f(y) \leq f(x) + \nabla f(x)^T(y-x) + \frac{\mu}{2}\|y - x\|^2, \quad \forall x,y\in \mathbb{R}^d.
\end{equation*}
\label{assmp4}
\end{assumption}

\subsection{Convergence Analysis Results}
\label{sec:IV-B}
Here we provide the theoretical bounds for our algorithm in both non-convex and strongly convex settings without assumptions on data-heterogeneity. We define $\overline{x}_g^t = \mathbb{E}_{h_S}[x_g^t]$ to be the server model's expectation over the sample sets $h_S$.

\begin{theorem}\textbf{(Non-convex)}
Under Assumptions \ref{asmp1}, \ref{asmp2}, and \ref{assmp3}. Let $\beta_{s} = \frac{m_{s} - h_{s}}{m_{s}}$ be the ratio of unsampled clients from each subnet, and define the \textbf{sample-wise mixing rate} term $\phi_{s} = (1 - \beta_{s}^2)$
. Define $p = \min(\phi_{1} , \ldots, \phi_{S})\in (0,1]$, and $q = \min(\rho_{1}, \ldots, \rho_{S})\in (0,1]$. With a choice of constant step size $\gamma < \mathcal{O}(\frac{p^2q}{KL})$, we have the following convergence rate:
{\begin{align}
    \textstyle \frac{1}{T}\sum_{t=1}^T\mathbb{E}\|\nabla f(\overline{x}_\mathrm{g} ^t)\|^2 
    \leq &\textstyle\mathcal{O}\bigg(\frac{\mathcal{H}^{1}}{TK\gamma} + \frac{LK\gamma^2\sigma^2}{2}     +\frac{L^2K^2\gamma^3\sigma^2}{p^4q^2}\bigg).
\end{align}}
\label{thm1}
\end{theorem}
\begin{proof}
If we define a Lyapunov function $\mathcal{H}^{t} = \mathbb{E}f(\overline{x}_\mathrm{g} ^t) - \mathbb{E}f(x^*) + c_0 K^3\gamma^3 \bigg(\frac{1}{p}Y^t + \frac{1}{q}Z^t\bigg) + c_1 \frac{K\gamma}{p}\Gamma^t$ and also combine Lemma \ref{dev_lem} with a constant $c_2$ (see definition of $\Gamma^t, Y^t, Z^t$, and $\Delta^t$ in the appendix), we have:
{\begin{align*}
0 \leq& \textstyle c_2 \gamma L^2\bigg(-\Delta^t + 3\rho_m^2K\frac{1}{n}\sum_{i=1}^n\mathbb{E}\|x_i^{t-1,K+1} - \overline{x}_\mathrm{g} ^t\|^2\notag\\
&\textstyle+ 6K^3\gamma^2 (4Y^t + 4Z^t+\mathbb{E}\|\nabla f(\overline{x}_\mathrm{g} ^t)\|^2) + 3K^2\gamma^2\sigma^2\bigg).
\end{align*}}
Choose the constants as the following: $c_2> 2$, 
$c_1 = 6L^2 c_2$, 
$c_0 = \frac{1152}{p^2}L^2c_2$, 
$\gamma < \frac{p^2q}{945KL}$.
Then by combining Lemmas \ref{noncvx_lemma}, \ref{dev_lem}, \ref{lem5}, \ref{lem6}, \ref{lem7}, and with values $D_1, D_2, D_3, D_4, D_5 \geq 0$ and $D > 0$, we can obtain the following recursive form:
{\begin{align*}
\textstyle\mathcal{H}^t - \mathcal{H}^{t-1} \leq&\textstyle -DK\gamma\mathbb{E}\|\nabla f(\overline{x}_\mathrm{g} ^{t-1})\|^2\notag\\
        &\textstyle-D_1Y^{t-1} - D_2Z^{t-1}- D_3\Gamma^{t-1} - D_4 \Delta^{t-1}\notag\\
        &\textstyle+\frac{D_5L^2}{p^4q^2 \cdot K}(K^3\gamma^3)\sigma^2 + \frac{L}{2K}(K^2\gamma^2)\sigma^2\\
        \leq&\textstyle -DK\gamma \mathbb{E}\|\nabla f(\overline{x}_\mathrm{g} ^{t-1})\|^2\notag\\
         &\textstyle+\frac{D_5L^2}{p^4q^2 \cdot K}(K^3\gamma^3)\sigma^2 + \frac{L}{2K}(K^2\gamma^2)\sigma^2.
\end{align*}}
 We can then unpack the recursion of the Lyapunov function and get the final rate.
\end{proof}
In Theorem \ref{thm1}, $p$ implies the sampling rate of the server and $q$ implies the information mixing ability of D2D communications. Large $p$ indicates the server samples a large amount of clients and a large $q$ indicates the connectivity of every subnet is high.
By carefully choosing the step size, we can achieve the following result:
\begin{corollary} Under the same conditions as in Theorem \ref{thm1}, there exists a constant step size such that:
\begin{align}
    &\textstyle\frac{1}{T}\sum_{t=1}^T\mathbb{E}\|\nabla f(\overline{x}_\mathrm{g} ^t)\|^2 
    = \mathcal{O}\bigg(\sqrt{\frac{\mathcal{H}^1\sigma^2L}{TK}} + (\frac{\mathcal{H}^1L\sigma}{KTp^2q})^{\frac{2}{3}} + \frac{\mathcal{H}^1L}{Tp^4q^2}\bigg).
\label{eq28}
\end{align}
\vspace{-0.15in}

\label{cor1}
\end{corollary}
\begin{proof}
    Follows from plugging Lemma \ref{unroll_lem} into Theorem \ref{thm1}.
\end{proof}
From observing \eqref{eq28}, the first two terms capture the effect of stochastic gradient variance on convergence. Choosing a larger $K$ decreases the effect of stochasticity on the convergence. We can also observe that the bound will be better when the values of $p$ and $q$ are large. A key observation is that the network topology that gives the tightest bound is under $p = q = 1$, which means the every time DS communication is performed, the server samples all clients and the topology of every in-subnet D2D communication is a fully connected graph.
By adding a stronger assumption on strong-convexity, we are able to get a stronger convergence result.
\begin{theorem}\textbf{(Strongly-convex)}
Under Assumptions \ref{asmp1}, \ref{asmp2},\ref{assmp3}, and \ref{assmp4}. Define the Lyapunov function $\mathcal{L}^t = \mathbb{E}\|\overline{x}_\mathrm{g} ^t - x^*\|^2 + \Gamma^t + \gamma Y^t + \gamma Z^t$ (see definition of $\Gamma^t, Y^t, Z^t$ in the appendix), then for a sufficiently small constant step size $\gamma$, we have
\begin{align}
    \textstyle\mathbb{E}\|\overline{x}_\mathrm{g} ^{T+1} \!- \!x^*\|^2 \leq (1 - \frac{\mu K\gamma}{4})^T\mathcal{L}^{1} + \|(I - A)^{-1}b\|_1,
\label{eq:24}
\end{align}
\label{thm2}
where $I$ is the identity matrix,
\small
\setlength{\arraycolsep}{1pt} 
\medmuskip = 1mu
%
%
\begin{gather*}
    \renewcommand*{\arraystretch}{1.5}
    A=\left(1-\frac{p}{2}\right)I +
    \gamma KL\begin{bmatrix}
        \frac{p-\mu \gamma K}{2\gamma KL} & 9(1 \!-\! p) & 72K^2\gamma & 72K^2\gamma\\
        \frac{14\gamma KL}{p} &  \frac{36\gamma L}{p} & \frac{14K}{pL} & \frac{14K}{pL} \\
        \frac{72\gamma^2K^2L^3}{p} & \frac{30L}{p} & \frac{240K^2\gamma L}{p} & \frac{240K^2\gamma L}{p}\\
        \frac{168\gamma^2K^2L^3}{q} & \frac{78L}{q} & \frac{624K^2\gamma L}{q} & \frac{624 K^2\gamma L}{q}\\
    \end{bmatrix},
\end{gather*}
\begin{gather*}
\renewcommand*{\arraystretch}{1.5}
    b = \begin{bmatrix}
        \frac{2\gamma^2K}{n} + 9K^2\gamma^3L\\
        K\gamma^2 + 3K^3\gamma^4L^2\\
        \frac{2\gamma}{qK} + \frac{30K^3\gamma^3L^2}{q}\\
        \frac{2\gamma}{qK} + \frac{78K^3\gamma^3L^2}{q}\\
    \end{bmatrix}\sigma^2.
\end{gather*}
\normalsize 

\end{theorem}
\begin{proof}
    By combining Lemmas \ref{scvx_lemma}, \ref{dev_lem}, \ref{lem5}, \ref{lem6}, \ref{lem7}, we can form the recursive form:
{\small    \begin{equation*}
        \textstyle\begin{bmatrix}
            \mathbb{E}\|\overline{x}_\mathrm{g} ^t \!-\! x^*\|^2\\
            \Gamma^t\\
            \gamma Y^t\\
            \gamma Z^t
        \end{bmatrix}
        \leq         
        A\begin{bmatrix}
            \mathbb{E}\|\overline{x}_\mathrm{g} ^{t-1} \!-\! x^*\|^2\\
            \Gamma^{t-1}\\
            \gamma Y^{t-1}\\
            \gamma Z^{t-1}
        \end{bmatrix} + b.
    \end{equation*}}
If we choose the our step size with the following inequality:
\begin{align}
    \textstyle\gamma \leq \min\Big(&\textstyle\frac{\min(p,q)\mu}{K(14L^2+240L^3)},\frac{1}{18KL}, \frac{4}{\mu K},\textstyle\frac{\min(p, q)p}{2(45KL + 108KL^2 + K\mu/4)},\nonumber\\&\textstyle\frac{\min(p,q)^2}{2(86K^2+864K^2L+K\mu/4)}\Big),
    \label{eq:14}
\end{align}
then we have $\rho(A) \leq \|A\|_1 \leq 1 - \frac{\mu K\gamma}{4} < 1$. Unrolling the stochastic noise part $\sum_{t=0}^{T-1}A^tb \leq (I - A)^{-1}b$, we have:
{\begin{align*}
\textstyle\mathcal{L}^{T+1} \leq (1 - \frac{\mu K\gamma}{4})^T\mathcal{L}^{1} + \|(I - A)^{-1}b\|_1.
\end{align*}}
Lower bound $\mathcal{L}^{T+1}$ with $ \mathbb{E}\|x_\mathrm{g} ^{T+1} - x^*\|^2$,
then the proof is complete.
\end{proof}
The second term on the RHS of \eqref{eq:24} is related to stochastic gradient noise $\sigma^2$, which can be controlled by the step size. By choosing a specific step size, we obtain the following result.
\begin{corollary}
    Under the same conditions as in Theorem \ref{thm2}, if we define $\zeta_0 = \mathbb{E}\|\overline{x}_\mathrm{g} ^1 - x^*\|^2 + \Gamma^1$, $\zeta_1 = Y^1 + Z^1$, and the RHS of \eqref{eq:14} to be $\overline{\gamma}$, then there exists a constant step size $\gamma = \min(\overline{\gamma}, \frac{\ln(\max(1, \mu K(\zeta_0 + \Bar{\gamma} \zeta_1)T/\sigma^2))}{\mu K T})$ such that:
{\begin{align}
    \textstyle\mathbb{E}\|\overline{x}_\mathrm{g} ^{T+1} - x^*\|^2 \leq&\textstyle \Tilde{\mathcal{O}}\bigg(\exp( - pqT)\cdot(\zeta_0 + \zeta_1)\bigg) \notag\\
    &\textstyle+ \mathcal{O}\left(\frac{\sigma^2}{\mu KT}\right) + \Tilde{\mathcal{O}}\left(\frac{L^5\sigma^2}{\mu^5T^5pq}\right).
\label{eq:cor2}
\end{align}}
\end{corollary} 
\begin{proof}
    Starting from \eqref{eq:24}, we find see that:
{\small\begin{align*}
    \textstyle\mathbb{E}\|\overline{x}_\mathrm{g} ^{T+1} - x^*\|^2 &\textstyle\leq (1 - \frac{\mu K\gamma}{4})^T\mathcal{L}^{1} + \|(I - A)^{-1}b\|_1\\
    &\textstyle\leq \exp( - \frac{\mu K\gamma}{2}T)\mathcal{L}^1 + \mathcal{O}\left(\frac{K^5L^5\sigma^2\gamma^5}{pq}\right).
\end{align*}}
If we define the RHS of \eqref{eq:14} to be $\overline{\gamma} = \mathcal{O}(\frac{\mu pq}{L^3K})$ and choose $\gamma = \min(\overline{\gamma}, \frac{\ln(\max(1, \mu K(\zeta_0 + \Bar{\gamma} \zeta_1)T/\sigma^2))}{\mu K T})$, then we have:
{\small\begin{align*}
        \mathbb{E}\|\overline{x}_\mathrm{g} ^{T+1} - x^*\|^2 \leq& \mathcal{O}\bigg(\exp( - \frac{\mu K\overline{\gamma}}{2}T)(\zeta_0 + \overline{\gamma}\zeta_1)\bigg)\\
        &+ \mathcal{O}\left(\frac{\sigma^2}{\mu KT}\right) + \Tilde{\mathcal{O}}\left(\frac{L^5\sigma^2}{\mu^5T^5pq}\right).
\end{align*}}
Plug in $\overline{\gamma}$ 
and the proof is complete.
\end{proof}
With the strong-convexity assumption, the theorem demonstrates that our algorithm has a linear convergence rate under $\sigma^2 = 0$. The terms $Y^t$ and $Z^t$ in both Lyapunov functions capture how the two gradient tracking terms $y^t$ and $z^t$ approximate the gradient information in the network, respectively. The term $\Gamma^t$ tracks the distance between the aggregated model during DS communication and the local models before aggregation. In both corollaries
, we observe the that these three terms converge to zero along with $\frac{1}{T}\sum_{t=1}^T\mathbb{E}\|\nabla f(\overline{x}_\mathrm{g} ^t)\|^2$ and $\mathbb{E}\|\overline{x}_\mathrm{g} ^{T+1} \!-\! x^*\|^2$ when $T$ goes to infinity.


\subsection{Learning-Efficiency Co-Optimization}
\label{sec:IV-C}
Based on Corollary \ref{cor1}, we can see that the convergence speed is determined by $\frac{1}{p^4q^2}$. Without any additional conditions, the best choice to maximize the convergence speed is setting $p = q = 1$.
However, recall that the main objective of employing a semi-decentralized FL setting is to save on the communication costs between clients and the FL server, which is contradicted by choosing $p = 1$ (because the communication cost between devices and server would be maximized). We therefore propose a method for the central server to trade-off between convergence speed and communication cost, e.g., the energy or monetary cost for wireless bandwidth usage.

For each subnet $\mathcal{C}_s$, the server estimates a communication cost of $E_{s}$ for pulling and pushing variables from the subnet $\mathcal{C}_s$ and a communication cost of $E_{s}^{D2D}$ for every round of D2D communication performed by the subnet $\mathcal{C}_s$. As in the case of Theorem \ref{thm1}, we define $\beta_{s} = \frac{m_{s} - h_{s}}{m_{s}}$ to be the ratio of unsampled clients from each subnet. The server can solve the following minimization problem with a given the balance terms $\lambda_1,\lambda_2,\lambda_3,\lambda_4 > 0$.
{\begin{equation}
    \begin{aligned}
    \min_{\beta_{1}, \ldots, \beta_{S}, p, K}\quad &\textstyle \frac{1}{p^4} + \lambda_1 \sqrt{\frac{1}{K}} + \lambda_2(\frac{1}{Kp^2})^{\frac{2}{3}} \\
    &\textstyle+ \lambda_3\sum_{s=1}^S(1 - \beta_{s})\cdot E_{s}\\
    &\textstyle+ \lambda_4 K\sum_{s=1}^S E^{D2D}_{s},\\
    \textrm{subject to} \quad &\textstyle 0 \leq \beta_{s} \leq \frac{m_{s} - 1}{m_{s}},\quad 1\leq i \leq l,\\
    &\textstyle p = \min(1 - \beta_{1}^2, \ldots, 1 - \beta_{s}^2).
    \end{aligned}
    \label{eq16}
\end{equation}}
The ratio between $E^{D2D}_{s}$ and $E_{s}$ is set to  $E_{s}^{D2D}/E_{s} = \delta$
. A small $\delta$ means that the cost for a client to communicate with its neighbors using D2D communication is much cheaper than to perform a centralized aggregation.

Optimization \eqref{eq16} 
has a similar form to a geometric program (GP), with the exception of the last constraint $p = \min(1 - \beta_{1}^2, \ldots, 1 - \beta_{S}^2)$. By relaxing it to $p \leq \min(1 - \beta_{1}^2, \ldots, 1 - \beta_{S}^2)$, it becomes a posynomial inequality constraint which can be handled by GP~\cite{boyd2007tutorial}.
The solution of the relaxed problem will remain the same as \eqref{eq16}, which is provable via a straightforward contradiction. It is worth comparing this procedure to the more complex control algorithm proposed in \cite{lin2021semi}, where one must adaptively choose a smaller $K$ when client models in a subnet deviate from each other too quickly. With our GT methodology, the subnet drift is inherently controlled, allowing a constant $K$ throughout the whole training process.

\section{Numerical Evaluation}
\label{sec:V}


\begin{figure*}
\centerline{\includegraphics[width=\textwidth]{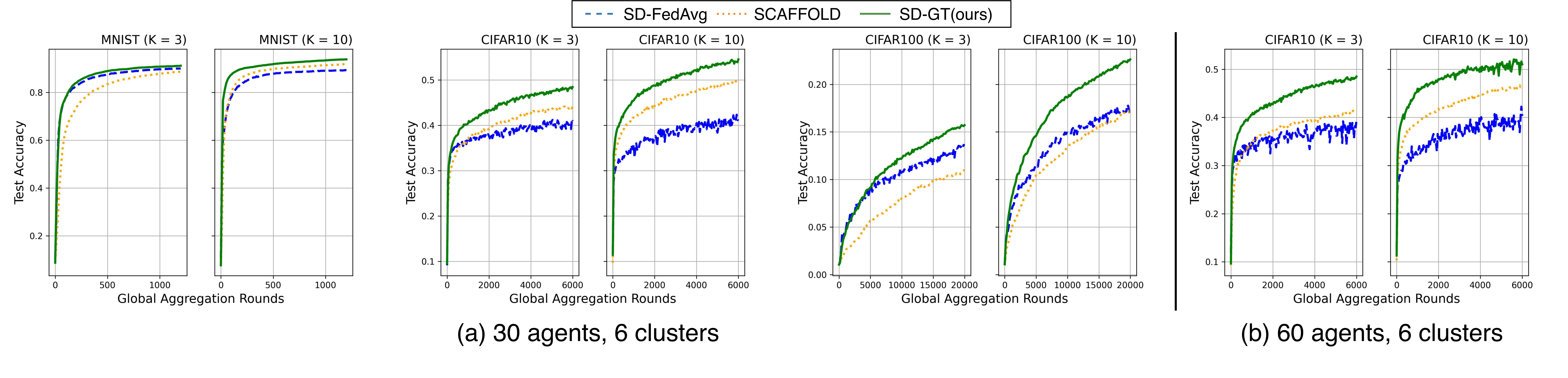}}\vspace{-0.15in}
\caption{Experimental results on real-world datasets ($\frac{h_{s}}{m_{s}} = 40\%$). By fixing the sampling rate for each subnet to $40$ percent and observe the effect of performing multiple D2D communication rounds, we see that our algorithm SD-GT gains the most improvement from increasing the D2D rounds. The advantage of our method can still be observed even if we double the number of clients in the system.\vspace{-0.15in}}
\label{fig3}
\end{figure*}
\begin{figure}
    \centering    \centerline{\includegraphics[width=0.5\textwidth]{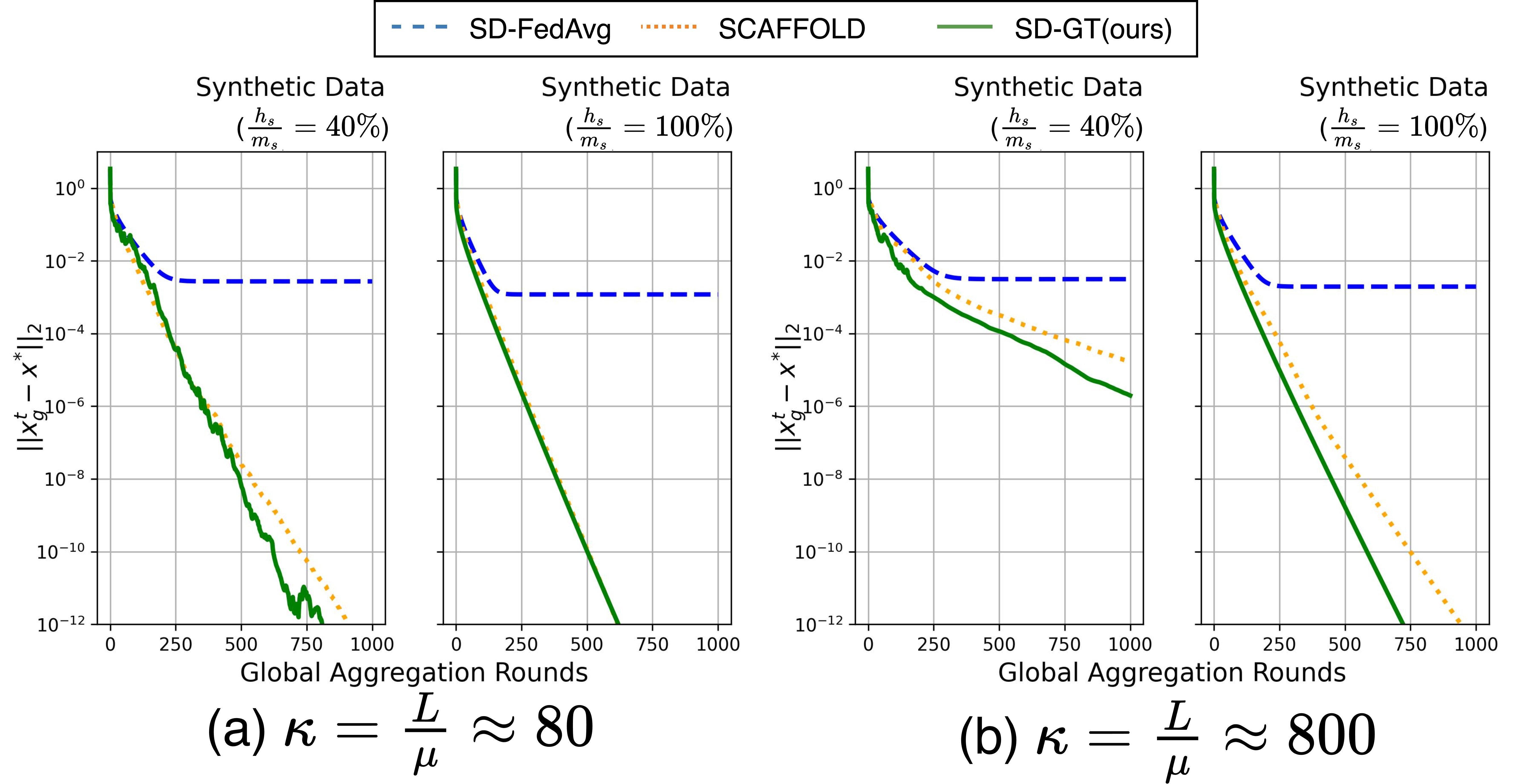}}
    \caption{Experimental results from synthetic dataset $(K = 40)$. By fixing the number of D2D rounds performed between two global aggregations, we see that our algorithm SD-GT is able to improve from increasing the sampling rate of each subnet while maintaining linear rate.\vspace{-0.15in}}
\label{fig4}
\end{figure}

\subsection{Experimental Setup}

\textbf{System Setup.} We set the system to have $n = 30$ clients, the only exception being Figure \ref{fig3}b, where the system contains 60 clients. All system have $S = 6$ subnets, each with same number of clients. Each subnet's graph structure is a generated as a random geometric graph with radius between $0.5$ and $3.5$. The doubly stochastic weight matrices $W_1,...,W_S$ for each subnet are generated using the Metropolis-Hasting rule.

\textbf{Datasets.} There are in total four datasets that we use for testing. The first three are all image classification datasets, and last one one is a synthetic dataset.

\textit{(i) Real-world Datasets:} We first consider a neural network classification task using a cross entropy loss function on image datasets (MNIST\cite{lecun1998gradient}, CIFAR10, CIFAR100\cite{krizhevsky2009learning}). For all three datasets we use the training set for training and testing set for evaluation. Let $\mathcal{D}_i$ be the dataset allocated on client $i$; we set $|\mathcal{D}_i|$ to be the same $\forall i$. To simulate high data heterogeneity, all three datasets are partitioned in a non-i.i.d manner such that each client in the MNIST and CIFAR10 experiments only holds data of one class and each client in the CIFAR100 experiments only holds data from four classes.

\textit{(ii) Synthetic Dataset:} We consider a strongly convex Least Square (LS) problem using synthetic data. The reason that we show results on this dataset is to demonstrate the linear convergence of our algorithm under strong-convexity. In the LS problem, each client $i$ estimates an unknown signal $x_0 \in\mathbb{R}^d$ through linear measurements $b_i = A_ix_0 +
n_i$, where $A_i \in \mathbb{R}^{{|\mathcal{D}_i|}\times d}$ is the sensing matrix, and $n_i \in \mathbb{R}^{|\mathcal{D}_i|}$ is the additive noise. 
The value of $d$ is set to 200 and $x_0$ is a vector of i.i.d. random variables drawn from a normal distribution $\mathcal{N}(0,1)$. Each client $i$ receives $|\mathcal{D}_i| = 30$ observations. All additive noise is sampled from the same distribution $\mathcal{N}(0, 0.04)$. The sensing matrix $A_i$ is generated as follows: For the $j$\textsuperscript{th} row of the sensing matrix $a_{ij} \in \mathbb{R}^d$, let $z_1, ..., z_d$ be an i.i.d. sequence of $\mathcal{N}(0,1)$ variables, and fix some correlation parameter $\omega \in [0, 1)$. We initialize by setting the first element of the $j$\textsuperscript{th} row $a_{ij,1} = \frac{z_1}{\sqrt{1 - \omega^2}}$, and generate the remaining entries by applying the recursive update
$a_{ij,t+1} = \omega a_{ij,t} + z_{t+1}$, for $t = 1, 2, \ldots, d-1$.

\textbf{Baselines.} We consider two baselines for comparison with SD-GT.

\textit{(i) SD-FedAvg:} Our first baseline is a semi-decentralized version of FedAvg\cite{lin2021semi} (denoted as SD-FedAvg). Each client updates using only its local gradient, and communicates with its nearby neighbors within each subnet using D2D communication after every gradient computation. A global aggregation is conducted after every $K$ rounds of gradient computation. This baseline does not contain gradient tracking, thus allowing us to assess this component of our methodology.

\textit{(ii) SCAFFOLD:} We also run a comparison with SCAFFOLD\cite{karimireddy2020scaffold}, which is a centralized gradient tracking algorithm that samples the network for aggregation. This baseline does not contain semi-decentralized local model updates, which means while SD-FedAvg and SD-GT performs $K$ rounds of D2D communication and model update, SCAFFOLD computes $K$ rounds of local on-device updates. Comparing with SCAFFOLD allows us to assess the benefit of our methodology co-designing semi-decentralized updates with gradient tracking.


\begin{figure}
    \centering    \centerline{\includegraphics[width=0.5\textwidth]{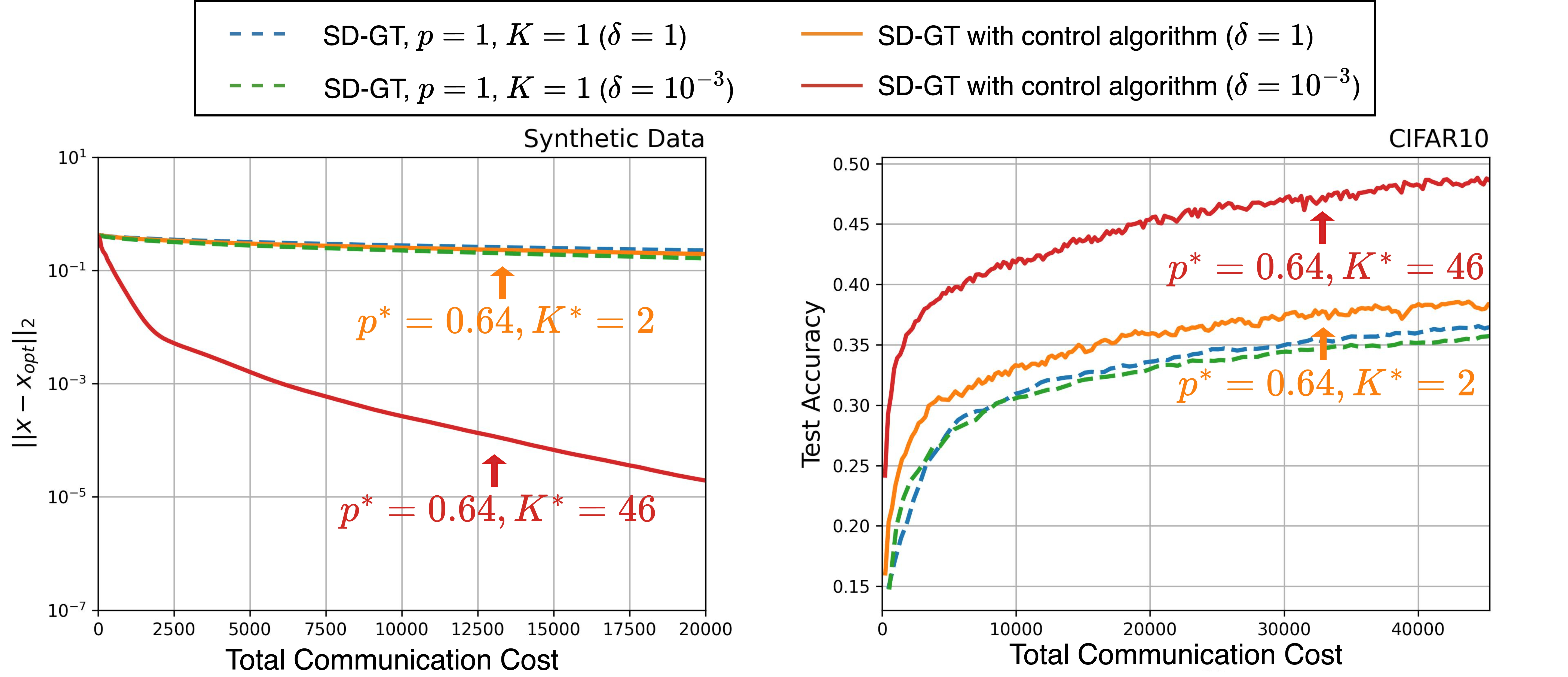}}
\caption{Experiments on the proposed co-optimization algorithm. When $\delta$ is small (D2D communication is cheap), the co-optimization algorithm is able to choose the sample rate and the number of D2D communication rounds such that more improvement is obtained using the same amount of communication.\vspace{-0.15in}}
\label{fig5}
\end{figure}

\textbf{Parameter Settings.} For experiments using the synthetic dataset, we tuned the value of $\omega$ to create two settings: (i) $\kappa \approx 80$, which is a simpler learning task, and (ii) $\kappa \approx 800$, which is more complicated. The step size for synthetic dataset experiments is set to $10^{-4}$. The step size is set to $\gamma = 1 \times10^{-2}$ for MNIST and $\gamma = 3 \times 10^{-3}$ for both CIFAR10 and CIFAR100. The neural network that we use for ML-tasks is a two layer fully-connected neural network with ReLU activations. We use deterministic gradients for the synthetic dataset and a batch-size of $512$ for the three ML datasets.

For experiments on our co-optimization algorithm, the communication cost $E_{s}$ for each subnet is a randomly drawn number between $1$ and $100$. The balance terms are set to $\lambda_1 = \lambda_2 = 1$,  $\lambda_3 = 10^{-1}$ and $\lambda_4 = 10^{-2}$. 

\subsection{Results and Discussion}
\textbf{Experiment 1: Convergence results for varying D2D rounds.} In Figure \ref{fig3}, we compare the model training performance of the algorithms on the image classification datasets. In these experiments, the central server samples $40\%$ of the clients from each subnet. We change the number of D2D rounds from $K = 3$ to $K = 10$ to observe the effect of performing multiple rounds of in-subnet consensus operations. In the MNIST and CIFAR10 experiments, we can see that SD-FedAvg struggles to gain improvement from increasing the number of D2D communications, while SD-GT is able to track the gradient information from other parts of the network and \textit{obtains better results when using larger number of D2D rounds}. Since in our setting we consider D2D communications to be relatively low cost compare to DS communications, we also see that \textit{the convergence of SCAFFOLD, which does not include D2D communication, is slower than SD-GT} even though both algorithms uses gradient tracking to correct the update direction. 

In experiments utilizing CIFAR100, unlike the results in MNIST and CIFAR10 experiments, SD-FedAvg is able to gain improvements from increasing the D2D rounds in the CIFAR100 experiments. This result is similar to how \cite{liu2023decentralized} compares between D-SGD and their method by running experiments on a evenly distributed MNIST and a non-i.i.d. MNIST. GT based methods do not gain much improvement over non-GT based methods when all clients are holding i.i.d. datasets, but gradient tracking is very important when dealing with networks with non-i.i.d. data. \textit{Our algorithm deals with data-heterogeneity better than algorithms that do not use gradient tracking, gaining more improvement from increasing the number of in-subnet D2D communication rounds.}

\textbf{Experiment 2: Convergence results for varying DS sample rates.} Figure \ref{fig4} compares the learning performance between the algorithms on the synthetic dataset. We set the number of D2D rounds to $K = 40$ for all experiments and compare between different sampling rates $\frac{h_s}{m_s}$. We see that SD-GT has a linear convergence to the global optimal solution under strong-convexity for any sampling rate and under different $\kappa$ values, which aligns with our theoretical results. When $\kappa$ is large, D2D communication improves the overall convergence speed, which gives our algorithm a faster linear convergence speed compared to SCAFFOLD, and SD-FedAvg converges to a radius away from the optimal solution related to the step size and the data-hetorogeneity assumption\cite{nemirovski2009robust, li2019convergence}. 
We also see that SCAFFOLD, while converging linearly, does so at a slower rate than SD-GT when $\kappa$ is large. The lack of D2D communication in SCAFFOLD makes the model aggregated from device to server contain gradient information of only a single device instead of the whole subnet. \textit{This emphasizes the performance benefits of combining D2D communications with gradient tracking in our algorithm when compared to the baselines.}

\textbf{Experiment 3: Co-optimization algorithm under different communication costs.} Figure \ref{fig5} evaluates the impact of {\tt SD-GT}'s learning-efficiency co-optimization procedure. Overall, we see that the optimization leads to improved SD-GT performance in terms of communication efficiency for the same learning quality. We compare the improvement gained from our co-optimization algorithm under a larger delta ($\delta=1$) and a smaller delta ($\delta = 10^{-3}$). When $\delta$ is large, which means that for each client, DS and D2D communications have similar cost, there is no need to perform multiple D2D communication rounds because the communication efficiency of using the co-optimization algorithm is similar to simply aggregating all clients at the central server after every gradient computation. However, when $\delta$ is small (i.e., in-subnet D2D communication is much cheaper than DS communication), the co-optimization algorithm's balancing between convergence speed and communication efficiency results in better convergence while incurring the same communication cost.

\section{Conclusion}
\noindent We developed SD-GT, the first gradient-tracking based semi-decentralized federated learning (FL) methodology. SD-GT incorporates dual gradient tracking terms to mitigate the subnet-drift challenge. We provided a convergence analysis of our algorithm under non-convex settings and strongly-convex settings, revealing conditions under which linear and sub-linear convergence rates are obtained. Based on our convergence results, we developed a low-complexity co-optimization algorithm that trades-off between learning quality and communication cost. Our subsequent experimental results demonstrated the improvements provided by SD-GT over baselines in SD-FL and gradient tracking literature.

\appendix

\subsection{Notation}
\label{Appendix:A}
For the proof, we define several matrix form notations. Let ${y^t} = [y_1^t, \ldots, y_n^t]$, ${z^t} = [z_1^t, \ldots, z_n^t]$, and ${x^{t,k}} = [x_1^{t,k},\ldots,x_n^{t,k}]$ be the matrix form of the variables. Let $\nabla F({x^{t,k}}) = [\nabla f_1(x_1^{t,k}), \ldots, \nabla f_n(x_n^{t,k})]$ be the matrix form of the gradient. Define $J = \frac{\mathbf{1_n} \mathbf{1_n}^T}{n}$ be the averaging matrix of all clients and $J_c = \mathrm{diag}(\frac{\mathbf{1_{m_{1}}} \mathbf{1_{m_{1}}}^T}{m_{1}},\ldots,\frac{\mathbf{1_{m_{S}}} \mathbf{1_{m_{S}}}^T}{m_{S}})$ to be the block diagonal matrix of averaging matrices for each subnet. Let $W = \mathrm{diag}(W_{1},\ldots, W_{S})$ be the block diagonal matrix of all weighted matrices for each subnet. Let $p$ and $q$ be constants that represents the mixing ability of inter-subnet and in-subnet (defined in Theorem \ref{thm1}). Define two block diagonal matrix $B = \mathrm{diag}(\beta_{1} I_{\mathcal{C}_1}, \ldots, \beta_{S} I_{\mathcal{C}_S})$ and $B^{'}= \mathrm{diag}((1- \beta_{1}) I_{1}, \ldots, (1-\beta_{S}) I_{\mathcal{C}_S})$ for Lemma \ref{lem5}.

We define $\Delta^t = \frac{1}{n}\sum_{i=1}^n\sum_{k=1}^K\mathbb{E}\|x_i^{t, k} - x_\mathrm{g} ^{t}\|^2$ to be the \textbf{client drift }term, $\Gamma^t = \frac{1}{n}\sum_{i=1}^n\mathbb{E}\|x_i^{t-1,K+1} - x_\mathrm{g} ^t\|^2$ be the \textbf{sampling error }term, and $Z^t = \frac{1}{n}\mathbb{E}\|z^t + \nabla F(x_\mathrm{g} ^t)(I - J_c)\|^2_F$, $Y^t = \frac{1}{n}\mathbb{E}\|{y^t} + \nabla F(x_\mathrm{g} ^t)(J_c - J)\|_F^2$ to be the \textbf{in-subnet} and \textbf{inter-subnet correction} terms, respectively.

\subsection{Supporting Lemmas}

\begin{lemma}(Unroll recursion lemma)
    For any parameters $r_0 \geq 0, b \geq 0, e\geq 0, u\geq 0$, there exists a constant step size $\gamma < \frac{1}{u}$ s.t. 
{\small\begin{equation}
    \textstyle\Psi_T := \frac{r_0}{T}\frac{1}{\gamma} + b\gamma + e\gamma^2 \leq 2\sqrt{\frac{br_0}{T}} + 2e^{\frac{1}{3}}(\frac{r_0}{T})^{\frac{2}{3}} + \frac{ur_0}{T}.
\end{equation}}
    \label{unroll_lem}
    \vspace{-0.30in}
\end{lemma}
\begin{proof}
    See Lemma C.5 in \cite{liu2023decentralized} or Lemma 15 in \cite{koloskova2020unified}.
\end{proof}
\begin{lemma}(Non-convex descent lemma) Under Assumption \ref{asmp1}, if we choose step size $\gamma < \frac{1}{4KL}$, then we have:
{\small\begin{equation}
    \begin{aligned}
    \textstyle\mathbb{E}f(\overline{x}_\mathrm{g} ^{t+1}) \leq &\textstyle \mathbb{E}f(\overline{x}_\mathrm{g} ^{t})-\frac{\gamma K}{4}\mathbb{E}\|\nabla f(\overline{x}_\mathrm{g} ^t)\|^2 \textstyle+ \frac{\gamma L^2}{n}\sum_{i,k}\mathbb{E}\|x_i^{t,k} - \overline{x}_\mathrm{g} ^t\|^2.
    \end{aligned}
\end{equation}}
\label{noncvx_lemma}
\vspace{-0.15in}
\end{lemma}
\begin{proof}
    By injecting L-smoothness, we can get:
{\small    \begin{equation*}
        \begin{aligned}
              \mathbb{E}f(\overline{x}_\mathrm{g} ^{t+1}) 
              \leq 
              &\textstyle \mathbb{E}f(\overline{x}_\mathrm{g} ^t) + \mathbb{E}\langle \nabla f(\overline{x}_\mathrm{g} ^t), \overline{x}_\mathrm{g} ^{t+1} - \overline{x}_\mathrm{g} ^t\rangle \textstyle+ \frac{L}{2}\mathbb{E}\|\overline{x}_\mathrm{g} ^{t+1}-\overline{x}_\mathrm{g} ^{t}\|^2\\
         \leq &\textstyle \mathbb{E}f(\overline{x}_\mathrm{g} ^{t}) + (L^2K^2\gamma^2-\frac{\gamma K}{2})\mathbb{E}\|\nabla f(\overline{x}_\mathrm{g} ^t)\|^2\\
         &\textstyle+ ( L^3K\gamma^2 + \frac{\gamma L^2}{2})\frac{1}{n}\sum_{i,k}\mathbb{E}\|x_i^{t,k} - \overline{x}_\mathrm{g} ^t\|^2\\
         \leq &\textstyle \mathbb{E}f(\overline{x}_g^{t})-\frac{\gamma K}{4}\mathbb{E}\|\nabla f(\overline{x}_g^t)\|^2 \\
        &\textstyle+ \gamma L^2\frac{1}{n}\sum_{i,k}\mathbb{E}\|x_i^{t,k} - \overline{x}_g^t\|^2.
        \end{aligned}
    \end{equation*}}
Let $\gamma < \frac{1}{KL}$ completes the proof.
\end{proof}
\begin{lemma}(Strongly-convex descent lemma)
    Under Assumptions \ref{asmp1} and \ref{asmp2}, if we choose $\gamma < \frac{1}{18KL}$, then we have:
{\small
    \begin{equation}
        \begin{aligned}
            \mathbb{E}\|\overline{x}_\mathrm{g} ^{t} - x^*\|^2 \leq& \big(1 - \frac{\mu K\gamma}{2}\big)\mathbb{E}\|\overline{x}_\mathrm{g} ^{t-1} - x^*\|^2 \\
            &+ 9\gamma KL(1 - p)\Gamma^{t-1} + 72K^3L\gamma^3 Y^{t-1} \\
            &+ 72K^3L\gamma^3Z^{t-1} + \frac{2\gamma^2K\sigma^2}{n} + 9K^2\gamma^3L\sigma^2.
        \end{aligned}
    \end{equation}}
    \label{scvx_lemma}
    \vspace{-0.15in}
\end{lemma}
\begin{proof}
    By injecting $L$-smoothness and $\mu$-strong-convexity at the same time, we can have:
    {\small\begin{equation*}
        \begin{aligned}
            \textstyle\mathbb{E}\|\overline{x}_\mathrm{g} ^{t} - x^*\|^2 =&\textstyle \mathbb{E}\|\overline{x}_\mathrm{g} ^{t-1} - \gamma \frac{1}{n} \sum_{i=1}^n\sum_{k=1}^K \nabla f_i(x_i^{t,k}, \xi^{t,k}_i)\textstyle - x^*\|^2\\
            \leq &\textstyle \big(1 - \frac{\mu K\gamma}{2}\big)\mathbb{E}\|\overline{x}_\mathrm{g} ^{t-1} - x^*\|^2 \\
            &\textstyle- 2\gamma K(\mathbb{E}f(\overline{x}_\mathrm{g} ^{t-1}) - f(x^*))\\
            &\textstyle+ 4\gamma^2K^2 \mathbb{E}\|\nabla f(\overline{x}_\mathrm{g} ^{t-1})\|^2 \\
            &\textstyle+ 3\gamma L\Delta^{t-1} + \frac{2\gamma^2K\sigma^2}{n}.\\
        \end{aligned}
    \end{equation*}}
    By plugging in Lemma \ref{dev_lem} and letting $\gamma < \frac{1}{18KL}$ completes the proof.
\end{proof}
\begin{lemma}(Deviation lemma) If we choose step size $\gamma < \frac{1}{8KL}$, then we have:
\begin{equation}
\begin{aligned}
    \textstyle\Delta^t
        \leq&     \textstyle3\rho_m^2K\frac{1}{n}\Gamma^t
        + 24K^3\gamma^2Y^r \\
        &    \textstyle+ 24K^3\gamma^2Z^r + 6K^3\gamma^2\mathbb{E}\|\nabla f(\overline{x}_\mathrm{g} ^t)\|^2 + 3K^2\gamma^2 \sigma^2.
\end{aligned}
\end{equation}
\vspace{-0.15in}
\label{dev_lem}
\end{lemma}
\begin{proof}
For any given D2D communication round $k$, simply plug in the iteration $x_i^{t, k} = \sum_{j \in \mathcal{N}^{\mathrm{in}}_i}w_{ij}(x_j^{t, k-1} - \gamma (\nabla f_j(x^{t,k-1}, \xi_j^{t, k-1})+ y_j^t + z_j^t))$:
\begin{equation}
    \begin{aligned}
        &\textstyle\frac{1}{n}\sum_{i=1}^n\mathbb{E}\|x_i^{t, k} - \overline{x}_\mathrm{g} ^{t}\|^2\\
        =&\textstyle \frac{1}{n}\sum_{i=1}^n\mathbb{E}\bigg\Vert\sum_{j \in \mathcal{N}^{\mathrm{in}}_i}w_{ij}\bigg(x_j^{t, k-1} - \gamma (\nabla f_j(x^{t,k-1}, \xi_j^{t, k-1})\\
        &\textstyle+ y_j^t + z_j^t)\bigg) - \overline{x}_\mathrm{g} ^{t}\bigg\Vert^2\\
        \leq&\textstyle (1 + \frac{1}{K-1} + 4K\gamma^2L^2)\frac{1}{n}\sum_{i=1}^n\mathbb{E}\|x_i^{t,k-1} - \overline{x}_\mathrm{g} ^t\|^2\\
        &\textstyle+ 8K\gamma^2Y^r+ 8K\gamma^2Z^r + 2K\gamma^2\mathbb{E}\|\nabla f(\overline{x}_\mathrm{g} ^t)\|^2 + \gamma^2 \sigma^2.
    \end{aligned}
    \label{eq41}
\end{equation}
    If we let $(1 + \frac{1}{K-1} + 4K\gamma^2L^2) = \alpha$ and $\gamma < \frac{1}{8KL}$, we can have $\alpha^{k-1} \leq \alpha^K \leq 3$ and $\sum_{k'=0}^{k-2}\alpha^{k'} \leq 3K$. Plugging them into the equation \eqref{eq41} completes the proof.
        
\end{proof}
\begin{lemma}(Inter-subnet GT lemma)
With constant step size $\gamma < \frac{1}{\sqrt{6}KL}$, we have: 
{\small
\begin{equation}
\begin{aligned}
        Y^t \leq \textstyle (1 - \frac{p}{2})Y^{t-1} \! + \!\frac{10L^2}{p}\Delta^{t-1} \!+\! \frac{12}{p}\gamma^2L^2K^2\mathbb{E}\|\nabla f(\overline{x}_\mathrm{g} ^{t-1})\|^2 \!+\! \frac{2\sigma^2}{pK}.
\end{aligned}
\end{equation}}
\vspace{-0.15in}
\label{lem5}
\end{lemma}
\begin{proof}
We have
{\small\begin{equation*}
    \begin{aligned}
        nY^t =&\textstyle \mathbb{E}\|{y^t} + \nabla F(\overline{x}_\mathrm{g} ^t)(J_c - J)\|_F^2\\
        =&\textstyle\mathbb{E}\| ({y^{t-1}} + \nabla F(\overline{x}_\mathrm{g} ^{t-1})(J_c - J)) B \\
        &\textstyle+ (\frac{1}{K}\sum_{k=1}^K\nabla F(x^{t-1,k}) - \nabla F(\overline{x}_\mathrm{g} ^{t-1}))(J - J_c) B' \\
        &\textstyle+ (\nabla F(\overline{x}_\mathrm{g} ^t) - \nabla F(\overline{x}_\mathrm{g} ^{t-1}))(J_c - J)\|_F^2 + n\frac{\sigma^2}{K}\\
        \leq&\textstyle (1 - \frac{p}{2})nY^{t-1} + \frac{6}{p}(\frac{(1 - \rho_m)^2nL^2}{K}\Delta^{t-1} \\
        &\textstyle+ 2\gamma^2nL^4K\Delta^{t-1} + 2\gamma^2L^2K^2n\mathbb{E}\|\nabla f(\overline{x}_\mathrm{g} ^{t-1})\|^2 )\\
        &\textstyle+\frac{6\gamma^2L^2Kn}{p}\sigma^2 +  n\frac{\sigma^2}{K}
    \end{aligned}
\end{equation*}}
Letting $\gamma < \frac{1}{\sqrt{6}KL}$ completes the proof.
\end{proof}
\begin{lemma}
    (In-subnet GT lemma) With constant step size $\gamma < \frac{1}{\sqrt{6}KL}$, we have: 

{\small    \begin{equation}
        \begin{aligned}
        Z^t &\leq\textstyle (1-\frac{q}{2})Z^{t-1} + \frac{26L^2}{q}\Delta^{t-1}
        \textstyle+ \frac{12K^2L^2\gamma^2}{q}\mathbb{E}\|\nabla f(\overline{x}_\mathrm{g} ^t)\|^2 + \frac{2\sigma^2}{qK}\\
        \end{aligned}
    \end{equation}}
    \label{lem6}
    \vspace{-0.15in}
\end{lemma}
\begin{proof}
We have:
{\small    \begin{equation*}
        \begin{aligned}
        nZ^t =&\textstyle \mathbb{E}\|{z^t} + \nabla F(\overline{x}_\mathrm{g} ^t)(I - J_c)\|^2_F\\
            =&\textstyle \mathbb{E}\|{z^{t-1}}W + \nabla F(\overline{x}_\mathrm{g} ^{t-1})(W - J_c) \\
            &\textstyle+ \frac{1}{K}\sum_{k=1}^K(\nabla F(\mathbf{x}^{t-1,k}) - \nabla F(\overline{x}_\mathrm{g} ^{t-1}))(W - I) \\
            &\textstyle+ (\nabla F(\overline{x}_\mathrm{g} ^t) - \nabla F(\overline{x}_\mathrm{g} ^{t-1}))(I - J_c)\|_F^2+ \frac{n\sigma^2}{K}\\
            \leq&\textstyle(1-\frac{p}{2})nZ^{t-1} + \frac{6}{p}(\frac{4nL^2}{K}\Delta^{t-1} \\
            &\textstyle+ \mathbb{E}\|\nabla F(x_g^t) - \nabla F(x_g^{t-1})\|_F^2) + \frac{\sigma^2}{K}\\
            \leq&\textstyle(1-\frac{q}{2})nZ^{t-1} + \frac{6n}{q}(\frac{4L^2}{K}\Delta^{t-1} + 2K\gamma^2L^4\Delta^{t-1} \\
            &\textstyle+ 2K^2L^2\gamma^2\mathbb{E}\|\nabla f(\overline{x}_\mathrm{g} ^t)\|^2 + L^2K\gamma^2n\sigma^2)+ \frac{n\sigma^2}{K}
        \end{aligned}
    \end{equation*}}
Letting $\gamma < \frac{1}{\sqrt{6}KL}$ completes the proof.
\end{proof}
\begin{lemma} (Sample Gap lemma)
{\small\begin{equation}
    \begin{aligned}
        \textstyle\Gamma^t
        \leq&\textstyle (1 - \frac{p}{2})\Gamma^{t-1} + \frac{12}{p}\gamma^2KL^2\Delta^{t-1} + \frac{12}{p}\gamma^2K^2Y^{t-1} \\
        &\textstyle+ \frac{12}{p}\gamma^2K^2Z^{t-1} + \frac{12}{p}\gamma^2K^2\mathbb{E}\|\nabla f(\overline{x}_\mathrm{g} ^{t-1})\|+ K\gamma^2\sigma^2
    \end{aligned}
\end{equation}}
\label{lem7}
\vspace{-0.15in}
\end{lemma}
\begin{proof}
We have:
    {\small\begin{equation*}
    \begin{aligned}
        \textstyle\Gamma^t =&\textstyle \frac{1}{n}\sum_{i=1}^n\mathbb{E}\|x_i^{t-1,K+1} - \overline{x}_\mathrm{g} ^t\|^2\\
        =&\textstyle \frac{1}{n}\sum_{i=1}^n\mathbb{E}\|x_i^{t-1,K+1} - \overline{x}_\mathrm{g} ^{t-1} \\
        &\textstyle+ \gamma\sum_{k=1}^K\frac{1}{n}\sum_{i=1}^n \nabla f_i(x_i^{t,k},\xi_i^{t,k})\|^2\\
        \leq&\textstyle (1 - \frac{p}{2})\Gamma^{t-1} + \frac{12}{p}\gamma^2KL^2\Delta^{t-1} + \frac{12}{p}\gamma^2K^2Y^{t-1}\\
        &\textstyle+ \frac{12}{p}\gamma^2K^2Z^{t-1} + \frac{12}{p}\gamma^2K^2\mathbb{E}\|\nabla f(\overline{x}_\mathrm{g} ^{t-1})\|^2 + K\gamma^2\sigma^2
    \end{aligned}
\end{equation*}}
\end{proof}

\clearpage  

\bibliographystyle{IEEEtran}
\bibliography{egbib}

\end{document}